\documentclass[11pt,a4paper]{article}
\usepackage[margin=1in]{geometry}

\usepackage[utf8]{inputenc}
\usepackage[T1]{fontenc}

\usepackage{amsmath}
\usepackage{amssymb}
\usepackage{amsthm}
\usepackage{mathtools}
\usepackage{thmtools}
\usepackage{isomath}

\usepackage{MnSymbol}
\usepackage{dsfont}
\usepackage[ruled,vlined,linesnumbered,noend]{algorithm2e}
\usepackage[textsize=footnotesize]{todonotes}
\usepackage{hyperref}
\usepackage{cite}
\usepackage[shortlabels]{enumitem}

\declaretheorem[numberwithin=section]{theorem}
\declaretheorem[numberlike=theorem]{lemma}

\declaretheorem[numberlike=theorem]{corollary}
\declaretheorem[numberlike=theorem]{definition}

\DeclareMathOperator{\Z}{\mathds Z}
\DeclareMathOperator{\R}{\mathds R}
\DeclareMathOperator{\N}{\mathds N}
\DeclareMathOperator{\argmax}{arg\ max}
\DeclareMathOperator{\argmin}{arg\ min}
\DeclareMathOperator{\Ot}{\tilde{\mathit{O}}} 
\DeclareMathOperator{\Omegat}{\tilde{\Omega}}

\DeclareMathOperator{\poly}{poly}

\let\P\relax
\DeclareMathOperator{\P}{\mathds{P}}
\DeclareMathOperator{\sign}{sign}
\DeclareMathOperator{\rank}{rank}
\DeclareMathOperator{\diag}{diag}

\SetKwFunction{LPSolve}{LPSolve}
\SetKwFunction{ComputeLeverageScores}{ComputeLeverageScores}
\SetKwFunction{ComputeApxWeights}{ComputeApxWeights}
\SetKwFunction{ComputeInitialWeights}{ComputeInitialWeights}
\SetKwFunction{PathFollowing}{PathFollowing}
\SetKwFunction{CenteringInexact}{CenteringInexact}
\SetKwFunction{ProjectMixedBall}{ProjectMixedBall}

\usepackage{graphicx}

\begin{document}

\title{Minimum Cost Flow in the CONGEST Model}

\author{Tijn de Vos\thanks{University of Salzburg, Austria. This work is supported by the Austrian Science Fund (FWF): P 32863-N. This project has received funding from the European Research Council (ERC) under the European Union's Horizon 2020 research and innovation programme (grant agreement No~947702).}}

\date{}

\maketitle              

\begin{abstract}
We consider the CONGEST model on a network with $n$ nodes, $m$ edges, diameter $D$, and integer costs and capacities bounded by $\poly n$. In this paper, we show how to find an exact solution to the minimum cost flow problem in $n^{1/2+o(1)}(\sqrt{n}+D)$ rounds, improving the state of the art algorithm with running time $m^{3/7+o(1)}(\sqrt nD^{1/4}+D)$~\cite{ForsterGLPSY21}, which only holds for the special case of unit capacity graphs. For certain graphs, we achieve even better results. In particular, for planar graphs, expander graphs, $n^{o(1)}$-genus graphs, $n^{o(1)}$-treewidth graphs, and excluded-minor graphs our algorithm takes $n^{1/2+o(1)}D$ rounds.
We obtain this result by combining recent results on Laplacian solvers in the CONGEST model~\cite{ForsterGLPSY21,Anagnostides0HZ22} with a CONGEST implementation of the LP solver of Lee and Sidford~\cite{LS14}, and finally show that we can round the approximate solution to an exact solution. Our algorithm solves certain linear programs, that generalize minimum cost flow, up to additive error~$\epsilon$ in $n^{1/2+o(1)}(\sqrt{n}+D)\log^3 (1/\epsilon)$ rounds.
\end{abstract}

\section{Introduction}

The CONGEST model~\cite{Peleg00} is one of the most widely studied distributed models. It consists of a network of $n$ nodes that communicate in synchronous rounds, where each node can exchange a message of size $O(\log n)$ with each of its neighbors. 
The minimum cost flow problem is considered one of the harder problems in the CONGEST model. Although the highest lower bound is $\Omegat(\sqrt n+D)$, which is the same as for `easier' problems such as shortest path, minimum spanning trees, bipartiteness, $s$-$t$ connectivity~\cite{PelegR00,Elkin06,SarmaHKKNPPW12}, it was only recently that the first distributed algorithm was presented~\cite{ForsterGLPSY21}. For the approximate version there exists some further, also quite recent, results~\cite{GhaffariKKLP18,BeckerFKL21}. These results use the powerful \emph{Laplacian paradigm} to obtain their results. 

The Laplacian paradigm encompasses a series of algorithms that combine numerical and combinatorial techniques. The \emph{Laplacian matrix} of a weighted graph $G$ is defined as $L(G):={\rm{Deg}}(G)-A(G)$, where ${\rm{Deg}}(G)$ is the diagonal weighted degree matrix: ${\rm{Deg}}(G)_{uu}:= \sum_{(u,v)\in E} w(u,v)$ and ${\rm{Deg}}(G)_{uv}:=0$ for $u\neq v$, and $A(G)$ is the adjacency matrix: $A(G)_{uv}:=w(u,v)$. This line of research was initiated by Spielman and Teng~\cite{SpielmanT04}, who showed that linear equations in the Laplacian matrix of a graph can be solved in near-linear time. More efficient sequential and parallel Laplacian solvers have been presented since~\cite{KOSZ13,KoutisMP14,KoutisMP11,CohenKMPPRX14,KyngS16,PengS14,KyngLPSS16}. The Laplacian paradigm has booked many successes, including but not limited to flow problems~\cite{Madry13,Sherman13,KelnerLOS14,Madry16,Peng16,CohenMSV17,LiuS2020FasterDivergence,LiuS20,AxiotisMV20}, bipartite matching~\cite{BrandLNPSS0W20}, and (parallel) shortest paths~\cite{Li20,AndoniSZ20}.

Recently, these developments have also made their way to the distributed world~\cite{GhaffariKKLP18,BeckerFKL21,ForsterGLPSY21,Anagnostides0HZ22,ForsterV22}. In particular, Forster, Goranci, Liu, Peng, Sun, and Ye~\cite{ForsterGLPSY21} provide a Laplacian solver that takes $n^{o(1)}(\sqrt{n}+D)$ rounds, which is near-optimal: they provide a $\tilde \Omega(\sqrt{n}+D)$ lower bound. Furtermore, they show that their Laplacian solver leads to an implementation of (minimum cost) maximum flow algorithms~\cite{Madry16,CohenMSV17} in the CONGEST model. In this paper, we significantly improve the round complexity of the algorithms solving the exact variants of these flow problems.

\subsection{Our Results}
Our main result is an algorithm that solves the minimum cost flow problem, so in particular also the maximum flow problem. 

\begin{restatable}{theorem}{thmflow}\label{thm:flow}
    There exists an algorithm that, given a directed graph $G=(V,E,w)$ with integer costs $q\in \Z_{>0}^m$ and capacities $c\in \Z_{>0}^m$ satisfying $||q||_\infty,||c||_\infty \leq M$, computes a minimum cost maximum $s$-$t$ flow in $\tilde O(\sqrt{n}T_{\rm{Laplacian}}(G)\log^3 M)$ rounds in the CONGEST model, where $T_{\rm{Laplacian}}(G)$ is the number of rounds needed to solve a Laplacian system on $G$.  
\end{restatable}

We know that $T_{\rm{Laplacian}}(G)= n^{o(1)}(\sqrt{n}+D)$ for general graphs~\cite{ForsterGLPSY21}, which is near-optimal. However, for certain graphs we can get better results. This is based on the concept of \emph{universally optimal} algorithms, which takes the topology of the input graph into account. The details regarding this can be found in \autoref{sc:prelim}. In particular, we have $T_{\rm{Laplacian}}(G)= n^{o(1)}D$ for planar graphs, expander graphs, $n^{o(1)}$-genus graphs, $n^{o(1)}$-treewidth graphs, and excluded-minor graphs.

Further we remark that Cohen, M\k{a}dry, Sankowski, and Vladu~\cite{CohenMSV17} show that the negative weight single source shortest path problem can be reduced to minimum cost flow and a non-negative weight shortest path computation. Using~\cite{ChechikM22} for the latter in $\tilde 
O(\sqrt{n}D^{1/4}+D)$ rounds, we obtain the following corollary. 
\begin{corollary}
    There exists an algorithm that, given a directed graph $G=(V,E,w)$ with integer weights $w\in \Z^m$  satisfying $||w||_\infty \leq M$, and source $s\in V$, computes shortest paths from $s$ in $\tilde O(\sqrt{n}T_{\rm{Laplacian}}(G)\log^3 M)$ rounds in the CONGEST model, where $T_{\rm{Laplacian}}(G)$ is the number of rounds needed to solve a Laplacian system on $G$. 
\end{corollary}

We obtain \autoref{thm:flow} by writing the problem as an LP, solving this LP up to high precision and rounding the result. Hereto we present an LP solver for certain linear programs in the CONGEST model. 

Formally, the setting is as follows. Let $A\in \R^{m\times n}$, $b\in\R^n$, $c\in\R^m$, $l_i\in \R\cup\{-\infty\}$, and $u_i\in \R\cup\{+\infty\}$ for all $i\in[m]$, where we assume $l_i\neq-\infty$ or $u_i\neq +\infty$. The linear program we want to solve is as follows
\begin{align*}
{\rm{OPT}}:= \min_{\substack{x\in\R^m : A^Tx=b \\ \forall i\in[m] : l_i\leq x_i\leq u_i}} c^Tx.
\end{align*}
We assume that the set of feasible solutions to the LP $\Omega^{\mathrm{o}}:=\{x\in\R^m : A^Tx=b,\ l_i\leq x_i\leq u_i\}$ is non-empty. 

\begin{restatable}{theorem}{thmLP}\label{thm:LP}
Let $A \in \R^{m\times n}$ be a constraint matrix with $\rank(A)=n$, let $b\in \R^n$ be a demand vector, and let $c\in \R^m$ be a cost vector. Moreover, let $x_0 \in \Omega^{\mathrm{o}}$ be a given initial point. Suppose a CONGEST network consists of $n$ nodes, where each node $i$ knows both every entire $j$-th row of $A$ for which $A_{ji}\neq 0$ and knows $(x_0)_j$ if $A_{ji}\neq 0$. Moreover, suppose that for every $y\in \R^n$ and positive diagonal $W\in\R^{m\times m}$ we can compute $(A^TWA)^{-1}y$ up to precision $\poly(1/m)$ in $T_{\rm{Laplacian}}(G)$ rounds. 
Let $U:=\max\{||1/(u-x_0)||_\infty,||1/(x_0-l)||_\infty,||u-l||_\infty,||c||_\infty\}$. Then with high probability the CONGEST algorithm \LPSolve outputs a vector $x\in \Omega^{\mathrm{o}}$ with $c^Tx \leq {\rm{OPT}} + \epsilon$ in $\Ot(\sqrt{n}\log^3(U/\epsilon)T_{\rm{Laplacian}}(G))$ rounds.
\end{restatable}

Intuitively, $U$ is a bound on the size of any constants or variables appearing in the LP. For graph problems, this is usually bounded by a polynomial in $n$ and $M$.

The formal statement of this theorem might seem somewhat convoluted; essentially it means that we can solve linear programs whose constraint matrix can be expressed in terms of the adjacency matrix, where each node knows the entries in the constraint matrix corresponding to its incident edges. Analogously, a node has to output the variables corresponding to its incident edges. 
This includes flow problems, see~\autoref{app:flow}. It also includes approximate fractional maximal matching. However, here the running time does not come close to the $O(\log(nM)/\epsilon^2)$ running time of Ahmadi, Kuhn, and Oshman~\cite{AhmadiKO18} (at least for $\epsilon=\Omega(1/n^{1/4})$).

\subsection{Related Work}
\paragraph{Distributed Flow Algorithms}
Our main point of reference is Forster, Goranci, Liu, Peng, Sun, and Ye~\cite{ForsterGLPSY21}. They provide the previous best minimum cost flow solver, which takes $m^{3/7+o(1)}(\sqrt n D^{1/4}+D)$ rounds\footnote{For simplicity, we restrict ourselves to graphs with weights bounded by $\poly(m)$ when discussing related work.}. Their approach uses the framework from Cohen, M\k{a}dry, Sankowski, and Vladu~\cite{CohenMSV17}, which uses $\tilde O(m^{3/7})$ iterations of another interior point method to solve a (different) LP representing the problem. We bring this number of iterations down to $\tilde O(n^{1/2})$. Moreover, this approach leads to an approximate solution that has to be made into an exact solution by running $\tilde O(m^{3/7})$ shortest path computations. Currently, the state of the art for algorithm for shortest path computations takes $\Ot(\sqrt n D^{1/4}+D)$ rounds~\cite{ChechikM22}, which is already (slightly) worse than the global (near optimal) round complexity for solving Laplacian systems. Moreover, this means that their set-up cannot benefit from the recent progress of (almost) universally optimal Laplacian solvers. Our approach solves the LP up to a higher precision, such that an internal rounding procedure gives the exact solution, and no further shortest path computations are necessary. 
A further improvement is that \cite{ForsterGLPSY21} only solves minimum cost flow in graphs with \emph{unit capacities}, where we solve it for arbitrary capacities. Further, \cite{ForsterGLPSY21} provides a maximum flow algorithm for graphs with arbitrary capacities, which takes $\tilde O(m^{3/7}U^{1/7}(n^{o(1)}(\sqrt n+D)+\sqrt n D^{1/4})+\sqrt m)$ rounds. This is the previous best result for maximum flow in the CONGEST model. 

For approximate versions, there exist some further results that only hold for undirected graphs. 
Ghaffari et al.~\cite{GhaffariKKLP18} give a $(1+\epsilon)$-approximate maximum flow in weighted undirected graphs in $n^{o(1)}(\sqrt n+D)/\epsilon^3$ rounds. Further, Becker et al.~\cite{BeckerFKL21} gave a $(1+\epsilon)$-approximation to unit capacity minimum cost flow in undirected graphs in $\Ot(n/\epsilon^2)$ rounds.

\paragraph{Interior Point Methods for Flow Problems}
The line of work giving solutions for flow problems through interior point methods is initiated by Daitch and Spielman~\cite{DS08}, who leverage the Laplacian solver of Spielman and Teng~\cite{SpielmanT04} in an $\Ot(m^{3/2})$ time algorithm. The most recent development is the near-linear time algorithm of Chen et al.~\cite{ChenKLPPS22}. However, their algorithm uses $\Omega(m)$ iterations, which seems to render it hard to implement it efficiently in a distributed setting, as any intuitive implementation uses at least one round per iteration. The algorithms with lowest iteration counts have either $\Theta(m^{3/7})$ iterations~\cite{Madry16,CohenMSV17}, or $\Theta(\sqrt{n})$ iterations\cite{LS14}. In our work, we show how to implement the latter efficiently in the CONGEST model.

\paragraph{Distributed Laplacian Solvers and Shortcut Quality} \label{sc:prelim}
Forster et al.~\cite{ForsterGLPSY21} provide a CONGEST model algorithm with $T_{\rm{Laplacian}}(G)= n^{o(1)}(\sqrt{n}+D)$, and show that this is existentially optimal. For any graph, we know that $T_{\rm{Laplacian}}(G)=\Omega(D)$, however it turns out that the $\sqrt{n}$-term is not necessary for every instance. To make this precise, we define the \emph{shortcut quality} of a graph, as introduced by Ghaffari and Haeupler~\cite{GhaffariH16}. Intuitively, the shortcut quality tells us how easy it is, given some partition of the nodes, to compute some simple function (e.g., a minimum over the values held by nodes) on each part separately. Since distributed algorithm design often has such functions at its core, the shortcut quality can be used both for better upper and lower bounds. 

\begin{definition}
    Let $G=(V,E)$ be an undirected graph whose node set $V$ is partitioned into $k$ disjoint subsets $V= P_1 \cupdot P_2 \cupdot \cdots \cupdot P_k$, such that each induced subgraph $G[P_i]$ is connected. A collection of $k$ subgraphs $H_1, \cdots, H_k$ is called a \emph{shortcut} of $G$ with congestion $c$ and dilation $d$ if
    \begin{enumerate}
        \item the (hop) diameter of $G[P_i]\cup H_i$ is at most $d$;
        \item every edge is included in at most $c$ graphs $H_i$.
    \end{enumerate}
    The \emph{quality} of the shortcut is defined as $c+d$. 
The \emph{shortcut quality} of $G$, denoted by $SQ(G)$, is defined as the smallest shortcut quality of the worst-case partition of $V$ into connected parts. 
\end{definition}

Anagnostides et al.~\cite{Anagnostides0HZ22} provide efficient algorithms for Laplacian solving in terms of the shortcut quality. Moreover, they provide an $\Omegat(SQ(G))$ lower bound. 

\begin{theorem}[\cite{Anagnostides0HZ22}]
    There exists a Laplacian solver with error $\epsilon>0$ in the CONGEST model that, given a graph $G$, takes $n^{o(1)}\poly(SQ(G))\log(1/\epsilon)$ rounds. In graphs with minor density $\delta$ and hop-diameter $D$, it takes $n^{o(1)}\delta D\log(1/\epsilon)$. 
\end{theorem}

Note that on graphs with minor density $n^{o(1)}$ the algorithm takes $n^{o(1)}D\log(1/\epsilon)$ rounds, matching the lower bound up to $n^{o(1)}$ factors. This includes planar graphs, $n^{o(1)}$-genus graphs, $n^{o(1)}$-treewidth graphs, and excluded-minor graphs. 

Further note that in particular on graphs $G$ with $SQ(G)=n^{o(1)}$ the algorithm takes $n^{o(1)}\log(1/\epsilon)$ rounds. This includes expanders, hop-constrained expanders, and the classes mentioned above with restricted diameter $D=n^{o(1)}$.


\section{Overview and Techniques}

\subsection{LP Solver}
For our LP solver, we give an implementation of Lee and Sidford's~\cite{LS14} LP solver in the CONGEST model. For correctness, we refer to~\cite{LS14}. In a similar fashion, Forster and de Vos~\cite{ForsterV22} gave an implementation of this algorithm in the Broadcast Congested Clique. 
In this distributed model, each round each node can send \emph{the same} $O(\log n)$-bit message to every other node in the network. This is in contrast to the CONGEST model, where nodes can send different messages, but only to its neighbors. Forster and de Vos essentially show that this LP framework uses $\sqrt n$ iterations, where each iteration involves:
\begin{enumerate}[(1)]
    \item Matrix-vector multiplication involving some matrix with entries corresponding to edges;
    \item Approximately solving a Laplacian system;\label{step:Laplacian_solve}
    \item Computing leverage scores; \label{step:comp_lev_scores}
    \item Projecting on a mixed norm ball.
\end{enumerate}

For this section (and the formulation of \autoref{thm:LP}), we assume that matrix-vector multiplication can be done efficiently. In practice, this means that we have to be able to write the constraint matrix in terms of the adjacency matrix. In other words, there should only be constraints that correspond to edges. 

For \ref{step:Laplacian_solve}, we can use \cite{ForsterGLPSY21,Anagnostides0HZ22}, as mentioned in \autoref{sc:prelim}. 

Concerning \ref{step:comp_lev_scores}; the leverage score of a matrix $M$ is defined by $\sigma(M) := \diag(M(M^TM)^{-1}M^T)$, where $\diag(\cdot)$ returns the diagonal vector. We  remark (similar to \cite{LS14,ForsterV22}) that, using the Johnson-Lindenstrauss Lemma, we can compute a sufficient approximation by some local sampling and a small number of matrix-vector multiplications and Laplacian solves. Details can be found in \autoref{app:LP}.

It remains to show that we can `project on a mixed norm ball'. The objective here is to project a vector $a\in \R^m$ onto a ball of mixed norm. In particular, given $l\in \R^m$, we consider the ball of mixed norm 1: $\mathcal B := \{x: ||x||_2+||l^{-1}x||_\infty\leq 1\}$. Now we need to compute $x\in \mathcal B$ closest to $a$, more formally we need to compute
\begin{align*}
\argmax_{||x||_2+||l^{-1}x||_\infty\leq 1} a^Tx.
\end{align*}
We do this by borrowing ideas from~\cite{LS14,ForsterV22}. Details can be found in \autoref{app:LP}.

\subsection{Minimum Cost Flow}
Let $G=(V,E)$ be a directed graph, with integer capacities $c\in \Z_{\geq 0}^m$, integer costs $q\in \Z_{\geq 0}^m$, and source and target nodes $s$ and $t$ respectively. The \emph{minimum cost (maximum) flow} problem is to find an $s$-$t$ flow of minimum cost, among all such flows of maximum value. More formally, we say that $f\in \R_{\geq 0}^E$ is a \emph{$s$-$t$ flow} if $f_e\leq c_e$ for all $e\in E$ and $ \sum_{e\in E : v\in e}f_e =0$. The \emph{value} of the flow is $\sum_{v\in V : (s,v)\in E}f_{(s,v)}$. The \emph{maximum flow}, is the flow of maximum value and the \emph{minimum cost (maximum) flow} is the flow of minimum cost $\sum_{e\in E}f_eq_e$ among all flows of maximum value. 

The minimum cost flow problem has a natural corresponding linear program. However, the state-of-the-art LP solvers only provide an approximate solution. To turn this efficiently into an exact solution, we do not consider the textbook LP formulation, but one that is closely related. After solving this up to high precision (additive error $\epsilon = O(1/\poly(m))$), we use the well-known fact that a minimum cost flow problem with integer input admits an optimal solution with integer values~\cite{KleinbergT06}, and we (internally) round the approximate fractional solution to an optimal integer solution. 

The technical contribution is to show that this particular LP formulation satisfies the demands of \autoref{thm:LP}. This is done in \autoref{app:flow}.
 
\section{A Distributed LP Solver}\label{app:LP}

In this section, we present our algorithm to solve a linear program, given a Laplacian solver. 
First, we reiterate the formal description of the problem. Let $A\in \R^{m\times n}$, $b\in\R^n$, $c\in\R^m$, $l_i\in \R\cup\{-\infty\}$, and $u_i\in \R\cup\{+\infty\}$ for all $i\in[m]$, where we assume $l_i\neq-\infty$ or $u_i\neq +\infty$. The linear program we try to solve is as follows
\begin{align*}
{\rm{OPT}}:= \min_{\substack{x\in\R^m : A^Tx=b \\ \forall i\in[m] : l_i\leq x_i\leq u_i}} c^Tx.
\end{align*}
We assume that the set of feasible solutions to the LP $\Omega^{\mathrm{o}}:=\{x\in\R^m : A^Tx=b,\ l_i\leq x_i\leq u_i\}$ is non-empty.

\thmLP*

The algorithm we provide in this section is an implementation of Lee and Sidford's~\cite{LS14} in the CONGEST model. We refer to them for the proof of correctness. For the two subroutines that we change, computing leverage scores and projecting on a mixed norm ball, we provide a correctness analysis. The remainder of this section consists of presenting the algorithm and proving the bound on the number of rounds. In both we follow notation of Forster and de Vos~\cite{ForsterV22}, who provide the equivalent for the Broadcast Congested Clique. 

Lee and Sidford~\cite{LS19}\footnote{In this section, we refer to the arXiv version~\cite{LS19} rather than the conference version~\cite{LS14}, whenever the technical details can only be found there.} show that it is sufficient to solve equations involving $A^TDA$ up to precision $\poly (1/m)$. We use this fact for our running time, and simplify our presentation by writing as if we solve such equations exactly. Similarly, we need to perform matrix-vector multiplication with the adjacency matrix and diagonal matrices only up to precision $\poly(\epsilon/(mU))$. Further we can assume all values are upper bounded by $\poly(mU/\epsilon)$. Due to the bandwidth constraint of the CONGEST model, these multiplications take $\Ot(\log(U/\epsilon))$ rounds. At the end of the computation both incident nodes to an edge know its value.

Further, we use throughout in runtime bounds that $D=\Ot(T_{\rm{Laplacian}}(G))$, which holds since $T_{\rm{Laplacian}}(G)=\Omegat(SQ(G))=\Omegat(D)$~\cite{Anagnostides0HZ22}.

\paragraph{Definitions and Set-Up}
On a high level, we perform a weighted path following interior point method. This means that throughout a number of iterations, given a current point $x^{(i)}\in \Omega^{\mathrm{o}}$, we find a point $x^{(i+1)}\in \Omega^{\mathrm{o}}$ closer to the optimal solution. To control that a point $x^{(i)}$ stays away from the boundary, we need to control $l_j\leq x_j^{(i)} \leq u_j$ for $j\in [n]$. This is done using a \emph{barrier function} $\phi_i(x_i)$, which goes to $\infty$ when $x_i$ goes to the boundary, i.e., to $l_j$ or $u_j$. The path then looks as follows:
\begin{align}
    x^{(i)} = \argmin_{A^Tx=b}\left( i\cdot c^Tx+\sum_{j\in[m]}\phi_j(x_j)\right).\label{eq:weighted_path}
\end{align}

To make this work, we need some more properties on $\phi$, leading to the definition of a self-concordant barrier function.

\begin{definition}\label{def:barrier_function}
	A convex, thrice continuously differentiable function $\phi\colon K\to \R^n$ is a $\nu$-\emph{self-concordant barrier function} for open convex set $K\subseteq \R^n$ if the following three conditions are satisfied
\begin{enumerate}
	\item $\lim_{i\to \infty} \phi(x_i) = \infty$ for all sequences $(x_i)_{i\in \N}$ with $x_i\in K$ converging to the boundary of $K$.
	\item $|\phi'''(x)[h,h,h]|\leq2|\phi''(x)[h,h]|^{3/2}$ for all $x\in K$ and $h\in\R^n$.
	\item $|\phi'(x)[h]|\leq \sqrt{\nu}|\phi''(x)[h,h]|^{1/2}$ for all $x\in K$ and $h\in\R^n$.
\end{enumerate}
\end{definition}

In our case, we choose $\phi$ as follows.
\begin{itemize}
	\item If $l_i$ is finite and $u_i=+\infty$, we use a log barrier: $\phi_i(x):=-\log(x-l_i)$.
	\item If $l_i=-\infty$ and $u_i$ is finite, we use a log barrier: $\phi_i(x):=-\log(u_i-x)$.
	\item If $l_i$ and $u_i$ are finite, we use a trigonometric barrier: $\phi_i(x):=-\log \cos(a_ix+b_i)$, where $a_i:=\tfrac{\pi}{u_i-l_i}$ and $b_i:=-\tfrac{\pi}{2}\tfrac{u_i+l_i}{u_i-l_i}$. 
\end{itemize}
This $\phi$ is a $1$-concordant barrier function~\cite{LS19}.
It can be computed internally in the CONGEST model, since we only require local knowledge of the constraints. By using this function in \autoref{eq:weighted_path}, we obtain a $\tilde O(\sqrt{m}\log(1/\epsilon))$ iteration method~\cite{renegar1988polynomial}. 

We can generalize \autoref{eq:weighted_path} to 
\begin{align}
    x^{(i)} = \argmin_{A^Tx=b}\left( i\cdot c^Tx+\sum_{j\in[m]}g_j(x)\phi_j(x_j)\right),\label{eq:weighted_path2}
\end{align}
for some \emph{weight functions} $g_j\colon \Omega^{\mathrm o}\to \R^m_{\geq 0}$. Lee and Sidford~\cite{LS14} show that using \emph{regularized Lewis weights} we only need $\tilde O(\sqrt{n}\log(1/\epsilon))$ iterations. 

To give the formal definition of the regularized Lewis weight function, we first introduce some general notation. 

\begin{itemize}
    \item For any matrix $M\in \R^{n\times n}$, we let $\diag(M)\in \R^n$ denote the diagonal of $M$, i.e., $\diag(M)_i:=M_{ii}$.
	\item For any vector $x\in \R^n$, we write upper case $X\in\R^{n\times n}$ for the diagonal matrix associated to $x$, i.e., $X_{ii}:=x_i$ and $X_{ij}:=0$ if $i\neq j$.
	\item For $x\in \Omega^{\mathrm{o}}$, we write $A_x:= (\Phi''(x))^{-1/2}A$.
	\item For $h\colon \R^n\to\R^m$ and $x\in\R^n$, we write $J_h(x)\in\R^{m\times n}$ for the Jaccobian of $h$ at $x$, i.e., $[J_h(x)]_{ij}:=\tfrac{\partial}{\partial x_j}h(x)_i$. 
	\item For positive $w\in \R^n_{>0}$, we let $||\cdot||_w$ the norm defined by $||x||_w^2 = \sum_{i\in[n]} w_ix_i^2$, and we let $||\cdot||_{w+\infty}$ the \emph{mixed norm} defined by $||x||_{w+\infty}=||x||_\infty + C_{\rm{norm}}||x||_w$ for some constant $C_{\rm{norm}}>0$ to be defined later. 
	\item Whenever we apply scalar operation to vectors, these operations are applied coordinate-wise, e.g., for $x,y\in \R^n$ we have $[x/y]_i:=x_i/y_i$, and $[x^{-1}]_i:=x_i^{-1}$.  
\end{itemize}

\begin{definition}
	A differentiable function $g\colon \Omega^{\mathrm{o}}\to \R^m_{>0}$ is a $(c_1,c_{\rm{s}},c_{\rm{k}})$-\emph{weight function} if the following bounds holds for all $x\in  \Omega^{\mathrm{o}}$ and $i\in[m]$:
\begin{itemize}
	\item size bound: $\max\{1,||g(x)||_1\}\leq c_1$;
	\item sensitivity bound: $e_i^TG(x)^{-1}A_x(A_x^TG(x)^{-1}A_x)^{-1}A_x^TG(x)^{-1}e_i\leq c_{\rm{s}}$; 
	\item consistency bound: $||G(x)^{-1}J_g(x)(\Phi''(x))^{-1/2}||_{g(x)+\infty}\leq 1-c_{\rm{k}} <1$. 
\end{itemize}
We denote $C_{\rm{norm}}:=24\sqrt{c_{\rm{s}}}c_{\rm{k}}$.
\end{definition}

Following Lee and Sidford~\cite{LS14}, we use the \emph{regularized Lewis weights}.
The $\ell_p$-Lewis weights generalize a $\ell_2$ measure of row importance called leverage scores. They are a key tool in approximating matrix $\ell_p$-norms. 

\begin{definition}
    For $M\in\R^{m\times n}$ with $\rank(M)=n$, we let $$\sigma(M) := \diag(M(M^TM)^{-1}M^T)$$ denote the \emph{leverage scores} of $M$. For all $p>0$, we define the \emph{$\ell_p$-Lewis weights} $w_p(M)$ as the unique vector $w\in R^m_{>0}$ such that $w=\sigma(W^{\tfrac{1}{2}-\tfrac{1}{p}}M)$, where $w=\diag(W)$. We define the \emph{regularized Lewis weights} as $g(x) := w_p(M_x)+c_0$, for $p=1-\tfrac{1}{\log(4m)}$ and $c_0 :=\tfrac{n}{2m}$. 
\end{definition}

The regularized Lewis weight function~$g$ is a $(c_1,c_{\rm{s}},c_{\rm{k}})$-weight function with $c_1 \leq \tfrac{3}{2}n$, $c_{\rm{s}}\leq4$, and $c_{\rm{k}}\leq2\log(4m)$ \cite{LS19}.

As said before, Lee and Sidford show that using \autoref{eq:weighted_path2} with this weight function for $\Ot(\sqrt n\log (1/\epsilon))$ iterations gives an $\epsilon$-approximate solution to our LP. 

\paragraph{Computing Leverage Scores}
Existing techniques for computing regularized Lewis weights compute leverage scores $ \sigma(M):= \diag(M(M^TM)^{-1}M^T)$ as an intermediate step.  As shown later, by repeatedly computing leverage scores, we can approximate the Lewis weights.  Unfortunately, there are no known efficient algorithms for computing leverage scores exactly. However, obtaining a sufficiently close approximation is feasible~\cite{SS11,DMMW12,Mahoney11,LMP13,Woodruff14,CLM+15}. We observe that $\sigma(M)_i = ||M(M^TM)^{-1}M^T e_i||_2^2$, and note that by the Johnson-Lindenstrauss lemma~\cite{JohnsonL1984extensions} this norm can be approximately preserved under projections onto a low dimensional subspace. In particular Achlioptas~\cite{Achlioptas03} gives an explicit (randomized) construction. 

\begin{theorem}[\cite{Achlioptas03}]\label{thm:JL}
    Let $m>0$ be an integer, let $\eta,\beta >0$ be parameters, let $k=\Omega(\beta\log m/\eta^2)$ be an integer,  and let $R\in \R^{k \times m}$ be a random matrix, where $R_{ij}=
    \pm 1/\sqrt k$, each with probability $1/2$. Then for any $x\in \R^m$ we have
    \[\P[ (1-\eta)||x||_2 \leq ||Rx||_2 \leq (1+\eta)||x||_2] \geq 1-m^{-\beta}.\]
\end{theorem}

Now we are ready to given an algorithm to compute $\sigma^{(\rm{apx})}$ such that $(1-\eta)\sigma(M)_i \leq \sigma^{(\rm{apx})}_i \leq (1+\eta)\sigma(M)_i$, for all $i\in[m]$.

\begin{algorithm}[H]
\SetAlgoLined \caption{\textsc{ComputeLeverageScores}($M,\eta$)}\label{alg:leverage_scores}
Set $k=\Theta(\log(m)/\eta^2)$.\\
Let $R\in \R^{k\times m}$ be a matrix where  $R_{ij}=\pm1/\sqrt{k}$, each with probability $1/2$. \\
Compute $p^{(j)}= M(M^TM)^{-1} M^T R^{(j)}$ for $j\in [k]$. \label{line:pj} \\
\Return{$\sum_{j=1}^{k}\left(p^{(j)}\right)^2$}.\label{line:returnpj}
\end{algorithm}

\begin{lemma}\label{lm:leverage_scores}
For any $\eta>0$, with probability at least $1-1/m^{O(1)}$ the CONGEST model algorithm \ComputeLeverageScores{$M,\eta$} computes $\sigma^{\rm{apx}}(M)$ such that 
\[ (1-\eta)\sigma(M)_i\leq \sigma^{\rm{apx}}(M)_i \leq (1+\eta)\sigma(M)_i\]
for all $i\in[m]$. If $M=WA$, for some diagonal $W\in \R^{m\times m}$, then it terminates in $\Ot((\log(U/\epsilon)+T_{\rm{Laplacian}}(G))/\eta^2)$~rounds. 
\end{lemma}
\begin{proof}
    \ComputeLeverageScores{$M,\eta$} returns $\sigma^{\rm{apx}}(M)_i$. Using that the matrix $M(M^TM)^{-1} M^T$ is symmetric, we obtain 
    \begin{align*}
        \sigma^{\rm{apx}}(M)_i &:= \sum_{j=1}^k (M(M^TM)^{-1} M^T R^{(j)})^2_i \\
        &= \sum_{j=1}^k (RM(M^TM)^{-1} M^T )^2_{ji}\\
        &= || R M(M^TM)^{-1} M^T e_i||_2^2.
    \end{align*}
    Since we also have $\sigma(M)_i = ||  M(M^TM)^{-1} M^T e_i||_2^2$, \autoref{thm:JL} gives us that 
    \[ (1-\eta)\sigma(M)_i\leq \sigma^{\rm{apx}}(M)_i \leq (1+\eta)\sigma(M)_i,\]
    with probability at least $1-1/m^{O(1)}$. Using a union bound, we can get the same guarantee for all $i \in [m]$ simultaneously.
    
    In the CONGEST model, we construct the required random matrix $R$ as follows. For each edge, the node with higher $\rm{ID}$ flips $k$ coins to determine the values $\pm1/\sqrt{k}$ and sends the result over the edge. This takes $O(k/\log n)=O(1/\eta^2)$ rounds. 

    For the computation in \autoref{line:pj}, we note that we can view this as $k$ times
    \begin{itemize}
        \item a matrix-vector multiplication $M^T R^{(j)}$, followed by
        \item a Laplacian system solve $(M^TM)^{-1} M^T R^{(j)}$, as $M^TM=A^TW^2A$, followed by
        \item a matrix-vector multiplication $M(M^TM)^{-1} M^T R^{(j)}$.
    \end{itemize}
    The first and last step can be done in $\Ot(\log(U/\epsilon))$ rounds, and the Laplacian solve can be done in $T_{\rm{Laplacian}}(G)$~ rounds. Finally, \autoref{line:returnpj} can be done internally. Hence we have total running time $\Ot((\log(U/\epsilon)+T_{\rm{Laplacian}}(G))/\eta^2)$.
\end{proof}

\paragraph{Computing the Weight Function}
 We continue by providing the algorithms for computing the initial weights, and for updating the weights throughout the path finding algorithm. As we use the latter for the former, we give the latter first.

\begin{algorithm}[H]
\SetAlgoLined \caption{\textsc{ComputeApxWeights}($M,p,w^{(0)},\eta$)}
$L=\max\{4,\tfrac{8}{p}\}$, $r=\tfrac{p^2(4-p)}{2^{20}}$, and $\delta=\tfrac{(4-p)\eta}{256}$.\\
$T= \left\lceil 80\left(\tfrac{p}{2}+\tfrac{2}{p}\right)\log\left(\tfrac{pn}{32\eta}\right)\right\rceil$. \\
\For{$j=1, \dots, T-1$}{
$\sigma^{(j)}=$\ComputeLeverageScores{$W^{\tfrac{1}{2}-\tfrac{1}{p}}_{(j)}M,\delta/2$}.\\
\For{$i\in [m]$}{
Let $w^{(j+1)}_i$ be the median of $(1-r)w^{(0)}_i$, $w^{(j)}_i-\tfrac{1}{L}\left(w^{(0)}_i-\tfrac{w^{(0)}_i}{w^{(j)}_i}\sigma^{(j)}_i\right)$, and $(1+r)w^{(0)}_i$.
}
}
\Return{$w^{(T)}$}.
\end{algorithm}

\begin{lemma}\label{lm:apx_wght}
    Let $W\in R^{m \times m}$ be some diagonal matrix, let $w^{(0)}\in \R_{>0}^m$ be a vector, and let $\eta\in (0,1]$ and $p\in [1-1/\log(4m),2]$ be parameters. Set $M=WA$. Then \textsc{ComputeApxWeights}($M,p,w^{(0)},\eta$) returns approximate weights in $$\Ot(\tfrac{\log(1/\eta)}{\eta^2}(\log(U/\epsilon)+T_{\rm{Laplacian(G)}}))$$ rounds.
\end{lemma}
\begin{proof}
    The algorithm consists of $T=\Ot((p+\tfrac{1}{p})\log(p/\eta))$ iterations. Using the assumption that $p\in [1-1/\log(4m),2]$, we get $T=\Ot(\log(1/\eta)$. In each iteration, we call the procedure \ComputeLeverageScores{$W^{\tfrac{1}{2}-\tfrac{1}{p}}_{(j)}M,\delta/2$} and compute some medians, the latter of which can be done internally. The call to \ComputeLeverageScores takes $\Ot((\log(U/\epsilon)+T_{\rm{Laplacian}}(G))/(\delta/2)^2)=\Ot((\log(U/\epsilon)+T_{\rm{Laplacian}}(G))/\eta^2)$ rounds, giving us the total running time as stated.
\end{proof}

For the properties and correctness of the approximate weights we refer to~\cite{LS19}.
Using the following algorithm, we compute the initial weights. We do this by iteratively bringing the all-ones vector closer to the initial weight vector. 

\begin{algorithm}[H]
\SetAlgoLined \caption{\textsc{ComputeInitialWeights}($A,p_{\rm{target}},\eta$)}
\SetKwInput{Input}{Input}
$p=2$.\\
$w=12c_{\rm{k}}\mathds{1}$.\\
\While{ $p\neq p_{\rm{target}}$\label{line:while}}{
$h=\tfrac{\min\{2,p\}}{\sqrt{n}\log\tfrac{me^2}{n}}\cdot r$.\\
Let $p^{(\rm{new})}$ be the median of $p-h$, $p_{\rm{target}}$, and $p+h$.\\
$w=$\ComputeApxWeights{$A,p^{(\rm{new})},w^{p^{(\rm{new})}/p},\tfrac{p^2(4-p)}{2^{22}}$}.\\
$p=p^{(\rm{new})}$.
}
\Return{\ComputeApxWeights{$A,p_{\rm{target}},w,\eta$}}.\label{line:return_init_w}
\end{algorithm}

\begin{lemma}\label{lm:init_wght}
    Let $\eta\in (0,1]$ and $p_{\rm{target}}\in [1-1/\log(4m),2]$ be parameters, then the CONGEST model algorithm \textsc{ComputeInitialWeights}($A,p_{\rm{target}},\eta$) returns initial weights in $\Ot((\sqrt{n}+\tfrac{\log(1/\eta)}{\eta^2}(\log(U/\epsilon)+T_{\rm{Laplacian(G)}}))$~rounds.
\end{lemma}
\begin{proof}
    The while loop of \autoref{line:while} finishes in $O(\sqrt{n}(p_{\rm{target}}+\tfrac{1}{p_{\rm{target}}})\log(m/n))$ iterations. Using that $p_{\rm{target}}\in [1-1/\log(4m),2]$, this simplifies to $\Ot(\sqrt{n})$ iterations. 
    Each iteration consists of internally computing $h$ and some medians, and a call to \ComputeApxWeights. This call requires precision $\tfrac{p^2(4-p)}{2^{22}}$, which is $\Omega(1)$ for our range of $p$. So the while loop takes $\Ot(\sqrt n(\log(U/\epsilon)+T_{\rm{Laplacian}}(G)))$ rounds in total. 

    Then in \autoref{line:return_init_w} we call \ComputeApxWeights with precision $\eta$, which takes $\Ot(\tfrac{\log(1/\eta)}{\eta^2}(\log(U/\epsilon)+T_{\rm{Laplacian(G)}}))$~rounds. Together this gives the stated running time. 
\end{proof}

For the properties and correctness of the initial weights we refer to~\cite{LS19}.

\paragraph{Algorithm}
In this section, we give the formal algorithm for solving the LP, together with a series of lemmas proving the running time of each subroutine.

\begin{algorithm}[H]
\SetAlgoLined \caption{\textsc{LPSolve}($x_0,\epsilon$)} \label{alg:LPSolve}
\SetKwInput{Input}{Input}
\Input{an initial point $x_0$ such that $A^Tx_0=b$.}
$w=$\ComputeInitialWeights{$A,1-1/\log(4m),\tfrac{1}{2^{16}\log^3 m}$}$+\tfrac{n}{2m}$, $d = -w\phi'(x_0)$. \\
$t_1=(2^{27}m^{3/2}U^2\log^4 m)^{-1}$, $t_2=\tfrac{2m}{\eta}$, $\eta_1=\tfrac{1}{2^{18}\log^3 m}$, and $\eta_2 =\tfrac{\epsilon}{8U^2}$.\\
$(x^{(\rm{new})},w^{(\rm{new})}) = $\PathFollowing{$x_0,w,1,t_1,\eta_1,d$}.\\
$(x^{(\rm{final})},w^{(\rm{final})}) = $\PathFollowing{$x^{(\rm{new})},w^{(\rm{new})},t_1,t_2,\eta_2,c$}. \\
\Return{$x^{(\rm{final})}$}.
\end{algorithm}

After computing the initial weights, this algorithm calls \PathFollowing twice, first to move the given initial point towards a central starting point with respect to the cost vector $c$, and second to move the path along from there. The algorithm \PathFollowing is as follows.  

\begin{algorithm}[H]
\SetAlgoLined \caption{\textsc{PathFollowing}($x,w,t_{\rm{start}},t_{\rm{end}},\eta,c$)}\label{alg:PathFollowing}
$t=t_{\rm{start}}$, $R=\tfrac{1}{768c_{\rm{k}}^2\log(36c_1c_{\rm{s}}c_{\rm{k}} m)}$, and $\alpha=\tfrac{R}{1600\sqrt{n}\log^2 m}$.\\
\While{$t\neq t_{\rm{end}}$\label{line:while_path}}{
$(x,w) = $\CenteringInexact{$x,w,t,c$}. \\
Let $t$ be the median of  $(1-\alpha)t$, $t_{\rm{end}}$, and $(1+\alpha)t$.
}
\For{$i=1,\dots,4c_{\rm{k}}\log(\tfrac{1}{\eta})$\label{line:for_path}}{
$(x,w)= $\CenteringInexact{$x,w,t_{\rm{end}},c$}.
}
\Return{$(x,w)$}.
\end{algorithm}

The progress steps in \PathFollowing are made by \CenteringInexact, which is as follows. 

\begin{algorithm}[H]
\SetAlgoLined \caption{\textsc{CenteringInexact}($x,w,t,c$)}\label{alg:CenteringInexact}
$R=\tfrac{1}{768c_{\rm{k}}^2 \log(36c_1c_{\rm{s}}c_{\rm{k}}m)}$, and $\eta=\tfrac{1}{2c_{\rm{k}}}$.\\
$\delta = \left|\left|P_{x,w}\left(\tfrac{tc+w\phi'(x)}{w\sqrt{\phi''(x)}}\right)\right|\right|_{w+\infty}$ \tcp{where $P_{x,w}:= I-W^{-1}A_x(A_x^TW^{-1}A_x)^{-1}A_x^T$.} 
$x^{(\rm{new})}= x- \tfrac{1}{\sqrt{\phi''(x)}} P_{x,w}\left(\tfrac{tc-w\phi'(x)}{w\sqrt{\phi''(x)}}\right)$.\\
$z =  \log\left(\ComputeApxWeights{$A_{x^{(\rm{new})}},1-1/\log(4m),w,e^R-1$}\right)$.\\
$u = \left(1-\tfrac{6}{7c_{\rm{k}}}\right)\delta\cdot$\ProjectMixedBall{$-\nabla\Phi\tfrac{\eta}{12R}(z-\log(w)),C_{\rm{norm}}\sqrt{w}$}.\\
$w^{(\rm{new})}=\exp(\log(w)+u)$.\\
\Return{$\left(x^{(\rm{new})},w^{(\rm{new})}\right)$}.
\end{algorithm}
We present the subroutine \ProjectMixedBall in \autoref{sc:mixednormball}. 
We prove the running times of these three algorithms in reverse order.

\begin{lemma}\label{lm:centering}
    The CONGEST model algorithm \CenteringInexact{$x,w,t,c$} terminates in $$\Ot(\log^2(U/\epsilon)T_{\rm{Laplacian}}(G))$$ rounds. 
\end{lemma}
\begin{proof}
    Computing  $P_{x,w}\left(\tfrac{tc+w\phi'(x)}{w\sqrt{\phi''(x)}}\right)$ takes $\Ot(T_{\rm{Laplacian}}(G))$ rounds, using internal computation for multiplying with diagonal matrices and a Laplacian solve. To compute $\delta$ and make it known to every node, we use $\Ot(D\log(U/\epsilon))$ rounds. 

    Next, we call \ComputeApxWeights with precision $\eta=\Omegat(1)$, so this takes $\Ot(\log(U/\epsilon)+T_{\rm{Laplacian}}(G))$ rounds by \autoref{lm:apx_wght}. Finally we call the algorithm \ProjectMixedBall, which takes $$\Ot(D\log^2(U/\epsilon))=\Ot(\log^2(U/\epsilon)T_{\rm{Laplacian}}(G))$$ rounds by \autoref{lm:mixed_norm_ball}.
\end{proof}

We use this result to prove the running time of \PathFollowing.

\begin{lemma}\label{lm:pathfollowing}
    Let $t_{\rm{start}},t_{\rm{end}}\geq 1$, and $\eta\in(0,1]$ be parameters. The CONGEST model algorithm \PathFollowing{$x,w,t_{\rm{start}},t_{\rm{end}},\eta,c$} terminates in $$\Ot(\sqrt n(|\log(t_{\rm{end}}/t_{\rm{start}})|+\log(1/\eta))\log^2(U/\epsilon)T_{\rm{Laplacian}}(G))$$ rounds. 
\end{lemma}
\begin{proof}
    First, we note that the while loop of \autoref{line:while_path} uses $\Ot(\sqrt n(|\log(t_{\rm{end}}/t_{\rm{start}})|+\log(1/\eta)))$ iterations~\cite{LS19}. Each such iteration consists of a call to \CenteringInexact and internal computations. Then the for loop of \autoref{line:for_path} takes $O(c_k\log(1/\eta))$ iterations, each consisting of a call to \CenteringInexact. Clearly this is dominated by the running time of the while loop. 

    Since \CenteringInexact takes $\Ot(\log^2(U/\epsilon)T_{\rm{Laplacian}}(G))$ rounds by \autoref{lm:centering}, we obtain the stated running time. 
\end{proof}

Finally, we give the running time of the complete algorithm.

\begin{lemma}\label{lm:LP}
    Given $\epsilon>0$, the CONGEST model algorithm \LPSolve{$x_0,\epsilon$} terminates in $\Ot(\sqrt{n}\log^3(U/\epsilon)T_{\rm{Laplacian}}(G))$ rounds. 
\end{lemma}
\begin{proof}
    Apart from some internal computation, this algorithm consists of three different parts: computing initial weight and two calls to \PathFollowing with different parameters. 

    The call to \ComputeInitialWeights takes $$\Ot(\sqrt{n}(\log(U/\epsilon)+T_{\rm{Laplacian(G)}}))$$ rounds, since we call it with precision $\tfrac{1}{2^{16} \log^3 m}$. 

    The execution of \PathFollowing{$x_0,w,1,t_1,\eta_1,d$} takes $$\Ot(\sqrt n\log(U)\log^2(U/\epsilon)T_{\rm{Laplacian}}(G))$$ rounds, by \autoref{lm:pathfollowing} and plugging in $t_1$ and $\eta_1$.

    The execution of \PathFollowing{$x^{(\rm{new})},w^{(\rm{new})},t_1,t_2,\eta_2,c$} takes $$\Ot(\sqrt n\log^3(U/\epsilon)T_{\rm{Laplacian}}(G))$$ rounds, by \autoref{lm:pathfollowing} and plugging in $t_1$, $t_2$ and $\eta_2$.

    The last running time dominates the first two and gives the stated result. 
\end{proof}

\section{Projecting on a Mixed Norm Ball}\label{sc:mixednormball}
In this section, we present a CONGEST algorithm for projecting on a mixed norm ball. This problem is defined as follows. Given $a,l\in \R^m$, find
\begin{align*}
\argmax_{||x||_2+||l^{-1}x||_\infty\leq 1} a^Tx.
\end{align*}
In the original work, Lee and Sidford~\cite{LS19} initially sort $m$ values and precompute $m$ functions on $a$ and~$l$. In the CONGEST model, these are expensive routines. We borrow ideas from Forster and de Vos~\cite{ForsterV22}, who overcame the same problem for the Broadcast Congested Clique. The rough idea is to only sort implicitly, and perform a binary search to reduce the number of functions that we have to compute to a manageable amount. We provide pseudocode in \autoref{alg:ProjectMixedBall}, with more details in the proof of \autoref{lm:mixed_norm_ball}. The pseudocode has a rather complicated binary search and some daunting equations in it. Both are probably best understood by examining the proof. 

Note that this problem has little to do with the graph structure in the CONGEST model, and as expected the algorithm actually does not make use of the graph structure other than establishing a shortest path tree for communication. 

\begin{algorithm}[H]
    \caption{\textsc{ProjectMixedBall}$(a,l)$}\label{alg:ProjectMixedBall}
    Determine the minimum value, maximum value, and step size of $\{|a_i|/l_i : i\in [m]\}$, denote this space of possible values $S$.\\
    For $s\in S$, let $i$ be the index of the value $|a_i|/l_i$ closest to $s$. \\
    Perform a binary search on $S$ w.r.t.\ $g_i$:\\
    \quad Compute $\sum_{k\in[j]}|a_k||l_k|$, $\sum_{k\in[j]} a_k^2$, and $\sum_{k\in[j]} l_k^2$ for $j\in\{i-1,i\}$.\\
    \quad Internally compute $$g_i:= \max_{t:i_t=i}t\sum_{k\in[i]}|a_k||l_k|+\sqrt{(1-t)^2-t^2\sum_{k\in[i]} l_k^2}\sqrt{||a||_2^2-\sum_{k\in[i]}a_k^2}.$$\\
    Let $t$ be the index corresponding to the maximal $g_i$. \\
    $x^{i}_j := \begin{cases} \tfrac{t}{1-t} \sign(a_j)l_j & \text{if } j\in[i] \\ \sqrt{\tfrac{1-\left(\tfrac{t}{1-t}\right)^2\sum_{k\in[i]}l_k^2}{||a||^2_2-\sum_{k\in [i]} a_k^2}a_j} & \text{otherwise. }\end{cases}$\\
    \Return{$x$}
\end{algorithm}

\newpage
\begin{lemma}\label{lm:mixed_norm_ball}
	Suppose the vectors $a,l\in \R^m$ are distributed over the network such that: 1) for each $i\in[m]$, $a_i$ and $l_i$ are known by exactly one node, 2) a node knows $a_i$ if and only if it knows $l_i$. Moreover, suppose that $||a||_\infty,||l||_\infty\leq O(\poly(m)U)$. Then the algorithm \ProjectMixedBall{$a,l$} finds 
\begin{align*}
\argmax_{||x||_2+||l^{-1}x||_\infty\leq 1} a^Tx
\end{align*}
up to precision $O(1/(\poly(mU/\epsilon))$ in $\Ot(D\log^2(U/\epsilon))$ rounds in the CONGEST model. 
\end{lemma}
\begin{proof}
    We rewrite the problem into maximizing over some concave function, which has a unique maximum that can be found using a binary search over the domain. We start by parameterizing the $\ell_2$-norm: 
    \begin{align*}
        \max_{||x||_2+||l^{-1}x||_\infty\leq 1} a^Tx &= \max_{0\leq t\leq 1} \left[ \max_{||x||_2\leq 1-t,\ -tl_i\leq x_i \leq tl_i} a^Tx \right]\\
        &= \max_{0\leq t\leq 1} (1-t)\left[ \max_{||x||_2\leq 1,\ -\tfrac{t}{1-t}l_i\leq x_i \leq \tfrac{t}{1-t}l_i} a^Tx \right]\\
        &= \max_{0\leq t\leq 1} g(t),
    \end{align*}
    where we define $g(t)$ as 
    \begin{align*}
        g(t) := (1-t)\left[ \max_{||x||_2\leq 1,\ -\tfrac{t}{1-t}l_i\leq x_i \leq \tfrac{t}{1-t}l_i} a^Tx \right].
    \end{align*} 
    We \emph{conceptually} sort the values of $a$ and $l$ with $|a_i|/l_i$ monotonically decreasing, i.e., we only sort them for this notation in the proof, the algorithm does not sort the values. Next, we write $i_t$ for the first coordinate $i\in [m]$ such that 
\begin{align*}
\tfrac{1-\left(\tfrac{t}{1-t}\right)^2\sum_{k\in[i_t]}l_k^2}{||a||^2_2-\sum_{k\in [i_t]} a_k^2} \leq \tfrac{\left(\tfrac{t}{1-t}\right)^2l_i^2}{a_i^2}.
\end{align*}
Now it can be shown (see e.g.~\cite{LS19}) that the vector that attains the maximum in $g(t)$ is $x^{i_t}\in \R^m$, defined by
\begin{align*}
x^{i_t}_j := \begin{cases} \tfrac{t}{1-t} \sign(a_j)l_j & \text{if } j\in[i_t] \\ \sqrt{\tfrac{1-\left(\tfrac{t}{1-t}\right)^2\sum_{k\in[i_t]}l_k^2}{||a||^2_2-\sum_{k\in [i_t]} a_k^2}a_j} & \text{otherwise. }\end{cases}
\end{align*}

We substitute this into the definition of $g(t)$: 
\begin{align*}
g(t) = t\sum_{k\in[i_t]}|a_k||l_k|+\sqrt{(1-t)^2-t^2\sum_{k\in[i_t]} l_k^2}\sqrt{||a||_2^2-\sum_{k\in[i_t]}a_k^2}.
\end{align*}
We note that $g(t)$ is a concave function (its second derivative is non-positive), hence it has a unique maximum. We find this maximum by searching over the domain. To do this, we rewrite $g$ in terms of the index $i_t$:
\begin{align*}
\max_{0\leq t\leq 1}g(t) &= \max_{0\leq t\leq 1} \max_{i\in[m]} g_i(t)\\
&=  \max_{i\in[m]} \max_{t : i_t = i} g_i(t),
\end{align*}
where 
\begin{align*}
g_i(t) := t\sum_{k\in[i]}|a_k||l_k|+\sqrt{(1-t)^2-t^2\sum_{k\in[i]} l_k^2}\sqrt{||a||_2^2-\sum_{k\in[i]}a_k^2}.
\end{align*}

Now fix a index $i$, and suppose a node knows $\sum_{k\in[j]}|a_k||l_k|$, $\sum_{k\in[j]} a_k^2$, and $\sum_{k\in[j]} l_k^2$ for $j\in\{i-1,i\}$. Then we can internally compute $g_i:=\max_{t : i_t = i} g_i(t)$, because we can internally find the range of $t$ where $i_t=i$, since we have that $i_t \geq i_s$ if $t\leq s$, hence the set of $t$ such that $i_t=j$ is an interval. 

Next, we describe how to compute the sums $\sum_{k\in[j]}$. We do this by constructing a shortest path tree of diameter $D$ from the node holding the values $a_j,l_j$. Along the tree, we aggregate the values of $|a_k||l_k|$, $a_k^2$, or $l_k^2$ respectively, for all $k\leq j$. The result can be broadcasted to all nodes without incurring extra costs. Note that if the indices are not known explicitly, the node holding $a_j$ and $l_j$, can first broadcast $|a_j|/l_j$, and then other nodes only add their values $|a_i||l_i|$ (and others respectively) if $|a_i|/l_i\leq |a_j|/l_j$. Since the values need to be maintained with precision $\poly(mU/\epsilon)$, sending one message needs at most $\Ot(\log(U/\epsilon))$ rounds, so the whole procedure takes at most $\Ot(D\log(U/\epsilon))$ rounds.

Naively, we would now be done by a simple binary search over $i\in [m]$, however we have the complication that we have only conceptually sorted the indices and hence nodes do not know which indices belong to the values they are holding. Instead we do a binary search over the possible values of $|a_i|/l_i$. Again using a communication tree from an arbitrary leader, we can find the global minimum, global maximum, and step size (least common multiple of denominators $l_i$) for the $|a_i|/l_i$ values. As not all values in the search space appear, we take the closest appearing value for a given value in the binary search. This gives a total search space of size $O(\poly(mU/\epsilon))$, so we need $\Ot(\log(U/\epsilon)$ iterations, each taking $\Ot(D\log(U/\epsilon))$ rounds.
\end{proof}

 \bibliographystyle{alpha}
 \bibliography{references}

\appendix
\section{Minimum Cost Flow}\label{app:flow}
In this section, we prove \autoref{thm:flow} by applying \autoref{thm:LP} to a suitable linear program and rounding the result to an exact solution accordingly. This particular LP formulation of minimum cost flow has first been presented by Daitch and Spielman~\cite{DS08}, and is used by Lee and Sidford~\cite{LS19}, and Forster and de Vos~\cite{ForsterV22}. As opposed to the formulation of \autoref{thm:flow}, we use $|V|$ and $|E|$ in this section to indicate the size of the node and edge set. We reserve $n$ and $m$ for the dimensions of the linear program, in line with \autoref{app:LP}. We write $M$ for the maximal edge capacity and cost.  

Let $B\in\R^{(|V|-1)\times |E|}$ be the edge-node incidence matrix with the row for the source~$s$ omitted. The variables of the LP consist of $x\in \R^{|E|}, y,z\in \R^{|V|}$ and $F\in \R$. The linear program is defined as follows. 

\begin{align*}
\min\ &\tilde{q}^T x+ \lambda(1^Ty+1^Tz)-2n\tilde{M}F\\
\text{\emph{subject to }} & Bx+y-z=Fe_t,\\
&0\leq x_i \leq c_i,\\
&0 \leq y_i\leq 4|V|M,\\
&0 \leq z_i\leq 4|V|M,\\
&0 \leq F\leq 2|V|M,\\
\end{align*}
where $\tilde{M}:=8|E|^2M^3$, $\lambda:= 28160 |E|^8 2M^9$, and $\tilde{q}=c+r$, where for each edge $r_e$ is a uniformly random number from $\left\{\tfrac{1}{4|E|^2M^2},\tfrac{2}{4|E|^2M^2}, \dots, \tfrac{2|E|M}{4|E|^2M^2}\right\}$. Daitch and Spielman~\cite{DS08} show that with probability at least $1/2$ this problem has a unique solution, which is also a valid solution to the original problem. After applying this reduction we (conceptually) scale everything by $4|E|^2M^2$ to ensure the cost vector is integral again. 

We set the variables as follows to obtain an initial interior point: $F=|V|M, x=\frac{c}{2}, y=2|V|M\mathds{1}-(B\frac{c}{2})^-+Fe_t, z=2|V|M\mathds{1}+(B\frac{c}{2})^+$, where we denote $a^+$ and $a^-$ for the vectors defined by 
\[ (a^+)_i := \begin{cases} a_i &\text{if }a_i\geq 0; \\ 0 &\text{else.}\end{cases} \hspace{3em}\text{and}\hspace{3em} (a^-)_i := \begin{cases} a_i &\text{if }a_i\leq 0 \\ 0 &\text{else.}\end{cases} \]
respectively. 

Next, we describe how we transform an $\epsilon$-approximate solution $x$ to this LP into an exact solution for the minimum cost flow problem. By introducing extra variables $y$ and $z$, we might have overshot the flow by at most $1^Ty+1^Tz\leq \epsilon$. To correct for this we set $\tilde x=(1-\epsilon)x$. We set $\epsilon:= \tfrac{1}{320|E|^4M^5}$, and then the error with respect to the unique solution is at most $1/6$~\cite{LS19}, so we have we can simply round the flow on each edge to the closest integer. Clearly both these steps can be done internally in the CONGEST model. 

To solve the above LP, we use \autoref{thm:LP} with $A=[B\ I\ -I\ -e_t]^T$. Actually, this does not use the entire network, but only $n=|V|-1$ nodes, since the source does not need to participate in the computation. Since the knowledge of the node-incident matrix $B$ is distributed as required, the knowledge of $A$ is distributed as required. The last step is to show that we can solve equations in $A^TWA$ in $T_{\rm{Laplacian}}$. This follows from~\cite{Gremban96} and is made explicit in~\cite{KOSZ13,ForsterV22}, who show that $A^TWA$ is symmetric and diagonally dominant, hence equations in $A^TWA$ can be solved by solving two Laplacian equations. We get $U/\epsilon=M\poly(|V|)$, so $\log^3 (U/\epsilon)=\Ot(\log^3 M)$. 

\end{document}